\documentclass[12pt]{l4dc2020} 

\title[On the Robustness of Data-Driven Controllers for Linear
Systems]{On the Robustness of Data-Driven Controllers for Linear
Systems}

\usepackage{times}

\usepackage{answers}
\usepackage{color}
\usepackage{setspace}
\usepackage{graphicx}
\usepackage{enumitem}
\usepackage{multicol}
\usepackage{mathrsfs}
\usepackage{array}
\usepackage{mathtools}
\usepackage{amsmath}
\usepackage{enumitem}

\newcommand{\R}{\mathbb{R}}

\newcommand{\transpose}{\mathsf{T}} 
\newcommand{\real}{\mathbb{R}}

\newcommand{\bP}{\mathbb{P}}

\newcommand{\wAcl}{\widetilde A_\mathrm{cl}}
\newcommand{\Acl}{A_\mathrm{cl}}
\newcommand{\tr}{\mathrm{tr}}
\newcommand{\wt}{\widetilde}
\newcommand{\ol}{\overline}
\newcommand{\ul}{\underline}

\newcommand{\setdef}[2]{\{#1 \; : \; #2\}}
\newcommand{\map}[3]{#1: #2 \rightarrow #3}
\newcommand{\subscr}[2]{{#1}_{\textup{#2}}}
\newcommand{\supscr}[2]{{#1}^{\textup{#2}}}

\newcommand\oprocendsymbol{\hbox{$\square$}}
\newcommand\oprocend{\relax\ifmmode\else\unskip\hfill\fi\oprocendsymbol}

\author{%
  \Name{Rajasekhar Anguluri} \Email{ranguluri@engr.ucr.edu}\\
  \Name{Abed AlRahman Al Makdah} \Email{aalmakdah@engr.ucr.edu}\\
  \Name{Vaibhav Katewa} \Email{vkatewa@engr.ucr.edu}\\
  \Name{Fabio Pasqualetti} \Email{fabiopas@engr.ucr.edu}\\[.3em]
  \addr Department of Mechanical Engineering, University of
  California at Riverside, Riverside, CA, 92507, USA.%
}

\graphicspath{{images/}}

\begin{document}

\maketitle

\begin{abstract}
  This paper proposes a new framework and several results to quantify
  the performance of data-driven state-feedback controllers for linear
  systems against targeted perturbations of the training data. We
  focus on the case where subsets of the training data are randomly
  corrupted by an adversary, and derive lower and upper bounds for the
  stability of the closed-loop system with compromised controller as a
  function of the perturbation statistics, size of the training data,
  sensitivity of the data-driven algorithm to perturbation of the
  training data, and properties of the nominal closed-loop system. Our
  stability and convergence bounds are probabilistic in nature, and
  rely on a first-order approximation of the data-driven procedure
  that designs the state-feedback controller, which can be computed
  directly using the training data. We illustrate our findings via
  multiple numerical studies.
\end{abstract}



\smallskip


\section{Introduction}
Data-driven algorithms are becoming increasingly more popular to solve
a variety of engineering problems, ranging from computer vision and
speech recognition to the design of stabilizing controllers for
dynamical systems (e.g., see \cite{PT-WLM-JG-ADA:17,
  BR:18}). While providing competitive performance under
nominal operating conditions and accurate data, these data-driven
methods typically offer no robustness guarantees against accidental or
adversarial manipulation of the training data, as demonstrated by
unfortunate incidents \citep{KP-MPM-DB-EB:18} and early studies
\citep{CDP-PT:19,SD-NM-BR-VY:19, AALM-VK-FP:19,AALM-VK-FP:19a}. This creates
concerns and poses critical limitations on the deployment of
data-driven control algorithms for practical problems.

In this paper we propose a novel framework and certain bounds to
characterize the robustness of data-driven state-feedback controllers
against perturbation of the training data. In particular, we view a
data-driven algorithm to design a stabilizing state-feedback
controller as a (differentiable) map from the collected data to the
space of controllers. Then, we compute a first-order approximation of
such map, which inherently measures the sensitivity of the data-driven
control algorithm to perturbations of its input data, and use it to
derive lower and upper bounds for the stability of the closed-loop
system with the controller obtained from the perturbed data. Our
stability results are probabilistic in nature, and they explicitly
depend upon the statistics of the perturbation, the size of the
training data of the data-driven algorithm, the sensitivity of the
data-driven algorithm, and the spectral properties of the nominal
closed-loop dynamics. Our results can be used to provide stability
guarantees for data-driven controllers, as well as to compare the
effectiveness of different data-driven procedures. Finally, we
illustrate our findings through a numerical example.

We will make use of the following notation. The cardinality of a set
$S$ is denoted by
$|S|$.  The spectral radius and trace of a square matrix are denoted
by $\rho(\cdot)$ and $\tr(\cdot)$, resp. The symbol
$\|\cdot\|$ denotes the Euclidean norm. The operators
vec($\cdot$) and
vec$^{-1}(\cdot)$ denote the vectorization and inverse vectorization
of a matrix and a vector, respectively.  The probability of an event
is denoted by
$\mathbb{P}(\cdot)$, and the expectation of a random variable is
represented by $\mathbb{E}[\cdot]$.
$\mathcal{N}(0,\Sigma)$ denotes a zero-mean Gaussian distribution with
covariance
$\Sigma$. The complementary cumulative distribution function (CDF) of
the standard normal distribution and the error function are denoted by
$\mathbb{Q}(\cdot)$ and $\text{erf}(\cdot)$.

\section{Problem Setup}\label{sec: problem_setup}
We consider a discrete-time linear time-invariant system given by
\begin{align}\label{eq: nominal_system}
  x(t + 1) &= A x(t) + B u (t), \qquad t\geq 0,
\end{align}
where $x\in \mathbb{R}^n$, $u \in \R^{m}$, $A \in \real^{n\times n}$, and
$B \in \real^{n \times m}$ denote, respectively, the state,
the input, the system matrix, and the input matrix. We assume that the system matrices $A$ and $B$ are unknown, and that a set of control experiments have been conducted to generate training data consisting of pairs of input sequences and samples of the system trajectories. Specifically, the training data is 
\begin{align}\label{eq: available_data}
  U = 
  \begin{bmatrix}
    u_1 & \cdots & u_N 
  \end{bmatrix} 
                   \in \R^{mT\times N}, 
                   \quad \text{ and } \quad 
                   X = C
                   \begin{bmatrix}
                     x_1 & \cdots & x_N
                   \end{bmatrix} \in \R^{p\times N}, 
\end{align}
where $N \in \mathbb{N}$ and $T \in \mathbb{N}$ denote the number and
length of the control experiments, respectively, $u_i$ and $x_i$ the
input and state trajectory of the $i$-th experiment, and
$C \in \real^{p \times nT}$ identifies the samples of the state
trajectories measured during each experiment. For instance, if
$C = [0 \; \cdots \; 0 \; I]$, then only the state at time $T$ is
measured during each experiment, as in \cite{GB-VK-FP:19}. Similarly,
if $C = I$, then the whole trajectory is measured as, for instance, in
\cite{CDP-PT:19}.

We assume that a static state-feedback data-driven controller $u= K x$
is used to stabilize the system \eqref{eq: nominal_system}, where
$K = F (U, X)$ and
$\map{F}{\real^{mT\times N} \times \real^{p\times N}}{\real^{m \times
    n}}$. The map $F$ denotes an arbitrary data-driven algorithm to
compute stabilizing controllers, such as the procedures described in
\cite{CDP-PT:19} and \cite{APV-AKS-FLL:19}. We make the following
assumptions:
\begin{itemize} 

\item[(A1)] The controller $K = F( U,X)$ stabilizes
  \eqref{eq: nominal_system}, that is, $\rho(A+BK)<1$.

\item[(A2)] The closed-loop matrix $A_{\text{cl}} = A+BK$ is diagonalizable.

\item[(A3)] The map $F (U, X)$ is Fr\'echet differentiable with
  respect to $X$, and admits a first-order Taylor expansion. Formally,
  for any $Z \in \R^{p \times N}$ the data-driven map satisfies
\begin{align}\label{eq: frechet_derivative}
\mathrm{vec}(F (U, X+Z)) = \mathrm{vec}(F( U, X)) + J_{X}( U, X) \mathrm{vec}(Z)+R(U, X, Z), 
\end{align}
with $\underset{\|Z\|\to 0}{ \lim}\frac{\|R(U,X,Z)\|}{\| Z\|}=0$,
where $J_X$ is the Jacobian matrix consisting of partial derivatives.

\end{itemize}
We remark that Assumption (A1) requires the data-driven algorithm to
stabilize the system with nominal data. Instead, Assumption (A2) is
convenient for the analysis and is not restrictive. Finally,
Assumption (A3) is a working assumption of this paper and is typically
used in similar studies.



The main objective of this paper is to quantify the robustness of the
data-driven controller $K = F(U, X)$ to perturbations of the
experimental data, which can be due, for instance, to measurement
noise or targeted adversarial manipulation of the data collection
sensors.\footnote{We focus on perturbations of $X$ only, although our
  methods can be extended to perturbations affecting both $U$
  and~$X$.} To this aim, let $\widetilde X = X + Z$ denote the data
perturbed by the zero-mean random noise $Z$. Let $\mathrm{supp}(Z)$
denote the set of compromised entries of $X$, that is,
$\mathrm{supp}(Z) = \setdef{i}{z_i \neq 0}$, with $z_i$ the $i$-th
component of $\mathrm{vec}(Z)$. Then, the main objective of this paper
is to characterize whether the perturbed closed-loop matrix
$\subscr{\widetilde A}{cl} = A + B \widetilde K$ is stable as a
function of the perturbation statistics, where
$\widetilde K = F(U, \widetilde X)$ denotes the data-driven controller
computed with the perturbed data. Notice that $Z$ is a random matrix,
and so are $\widetilde X$, $\widetilde K$, and
$\subscr{\widetilde A}{cl}$. Thus, the stability of
$\subscr{\widetilde A}{cl}$ will be studied in a probabilistic
framework.

\section{Robustness results for data-driven state-feedback
  controllers}\label{sec: robustness}
In this section we study the stability properties of the perturbed
closed-loop system $\subscr{\widetilde A}{cl} = A + B \widetilde
K$. In particular, we provide bounds for
$\mathbb{P}[\rho (\subscr{\widetilde A}{cl})\geq 1]$, which is
a well-defined random variable (see \cite[p.~85]{BA:73}) and quantifies
the probability that the closed loop system
$\subscr{\widetilde A}{cl}$ is unstable. 
We start with the following instrumental result to approximate the
matrix $\subscr{\widetilde A}{cl}$.

\begin{lemma}{\bf\emph{(First-order approximation of
      $\wAcl$)}}\label{prop: taylor_approximation}
  Let $\supscr{J}{v}_i$ denote the $i$-th column of $J_{X}( U,X)$ in
  \eqref{eq: frechet_derivative}, and let
  $J_i = \mathrm{vec}^{-1}(\supscr{J}{v}_i)$. Then, for any $\tau>0$,
  the perturbed closed-loop matrix satisfies
  \begin{align*}
    \lim_{\mathbb{E}[\|\mathrm{vec}(Z)\|] \rightarrow 0}
    \mathbb{P}\left[ \left\Vert \wt{A}_{\mathrm{cl}}- A_{\mathrm{cl}}
    - \textstyle \sum_{\mathrm{supp}(Z)} z_iB J_i \right\Vert \geq \tau
    \sqrt{\mathbb{E}[\| \mathrm{vec}(Z)\|]}\right] \to 0.
  \end{align*}
\end{lemma}
\begin{proof}
  From \eqref{eq: frechet_derivative} we have
  $\wt{K} = K + \textstyle\sum_{i=1}^{d} z_i
  \mathrm{vec}^{-1}(\supscr{J}{v}_i)+\mathrm{vec}^{-1}(R)$. Since
  $\ \wt{A}_{\mathrm{cl}} = A_{\mathrm{cl}}+B\widetilde K$, it now
  follows that
  $B\mathrm{vec}^{-1}(R)= \wt{A}_{\mathrm{cl}}- A_{\mathrm{cl}} +
  \textstyle \sum_{i=1}^{d} z_iB J_i$. By invoking
  \citep[Theorem.~3.1.1]{TK-DV:05}, we note that
  $\mathbb{P} \left[ \|B\mathrm{vec}^{-1}(R)\| \geq \tau\sqrt{
      \mathbb{E}[\|\mathrm{vec}(Z)\|]}\right]\to 0$ as
$\mathbb{E}[\|\mathrm{vec}(Z)\|] \to 0$.
\end{proof}

Lemma \ref{prop: taylor_approximation} states that, if the expected
norm of the perturbation $Z$ is sufficiently small, then $\wAcl$ can
be well approximated as
$A_{\mathrm{cl}} + \textstyle \sum_{\mathrm{supp}(Z)} z_iB J_i$. Thus, in what
follows we let
\begin{align}\label{eq: perturbed_system}
	\wt{A}_{\mathrm{cl}}=A_{\mathrm{cl}} + \textstyle
  \sum_{\mathrm{supp}(Z)} z_iB J_i .
\end{align}
The right hand term in \eqref{eq: perturbed_system} captures the
effect of each perturbation entry of $Z$ on the nominal system $\Acl$
in an additive form. In particular, the matrix $J_i$, consisting of
partial derivatives of the data-driven control map $F(U, X)$ with
respect to $X$, captures the sensitivity of the data-driven controller
to variations of the $i$-th component of $\mathrm{vec}(X)$. Also, the
specific form of the perturbation matrix allows us to capture the
effect of a particular subset of the data on the controller's
performance, since $z_i = 0$ if $i \not\in \mathrm{supp}(Z)$. Using
\eqref{eq: perturbed_system}, we now present bounds on the stability
of the closed-loop system for the case of normally distributed
perturbations. Our bounds make use of recent concentration inequality
results for the sum of random matrices \citep{JAT:15,SB-GL-PS:13}.




\begin{theorem}{\bf\emph{(Probabilistic bounds on the stability of
      $\wt{A}_{\mathrm{cl}}$)}}\label{thm: bound_probability_gaussian}
  Let $z_i \sim \mathcal{N}(0,\sigma_i^2)$, with $i \in
  \mathrm{supp}(Z)$. 
  Let
  $J_i$ be as in Lemma \ref{prop: taylor_approximation}, and
  define the following parameters:
  \begin{align*} 
    \overline v
    =\mathrm{max}\left\lbrace
       \left\| \sum\limits_{\mathrm{supp}(Z)} \sigma_i^2BJ_i
       \left(BJ_i\right)^\transpose\right\| ,
       \left\Vert \sum\limits_{\mathrm{supp}(Z)}\sigma^2_i\left(BJ_i\right)^\transpose
       BJ_i\right\Vert\right\rbrace, \text{ and } \underline
       v =\textstyle \sum\limits_{\mathrm{supp}(Z)}\sigma^2_i\left[\tr(BJ_i)\right]^2.
  \end{align*} 
  Let $\kappa=\|A_\mathrm{cl}\|\|A_\mathrm{cl}^{-1}\|$ be the condition number of $A_{\mathrm{cl}}$, and let
  $\mu=\tr(A_{\mathrm{cl}})$. Then,
  \begin{align}\label{eq: probability_bounds}
    \mathbb{Q}\left(\frac{n+\mu}{\sqrt{\ul
          v}}\right)+\mathbb{Q}\left(\frac{n-\mu}{\sqrt{\ul v}}\right)
    \leq \mathbb{P}\left[ \rho(\wt{A}_{\mathrm{cl}})\geq 1\right] \leq
    2n\exp\left( \frac{-\left(1-\rho(A_{\mathrm{cl}})\right)^2}{(2\ol v)
        \kappa^2 }\right). 
  \end{align} 
\end{theorem}
\begin{proof}
  Let $\Delta=\sum_{\mathrm{supp} (Z)} z_iBJ_i$. From \eqref{eq: perturbed_system} we have the following estimate:
  \begin{align}\label{eq: spectral_radius_bounds}
    n^{-1}|\mathrm{tr}(\wAcl)| \leq \rho(\wAcl)\leq \rho(\Acl) +\kappa\|\Delta\|.
  \end{align}
  The second inequality in \eqref{eq: spectral_radius_bounds}
  follows from the Bauer-Fike Theorem \cite[Chapter~4]{GWS-S:90}.
  Instead, the first inequality is trivially obtained using the
  triangle inequality.

  \textit{(Upper bound)} Let $t=(1-\rho(A_\text{cl}))/\kappa$ and
  $\widetilde \Delta=\sum_{\mathrm{supp} (Z)}\tilde z_i(\sigma_iBJ_i)$,
  where $\tilde z_i$ are independent and identically distributed
  random variables with zero mean and unit variance, which are also
  independent of the perturbation variables $z_i$. Notice the
  following chain of inequalities:
  \begin{align*}
    \bP[\rho(\widetilde A_\text{cl})\geq 1] \leq \bP\left[ \rho(\Acl)+\kappa\|\Delta\|\geq 1 \right] = \bP\left[\|\Delta\|\geq t \right]=\bP\left[\|\widetilde \Delta\|\geq t \right]\leq (2n) \exp({-t^2}/{2\ol v}).
  \end{align*}
  The first inequality follows by invoking monotonicity of probabilities on the set inclusion $\{\rho(\widetilde A_\text{cl})\geq 1\} \subseteq \left\lbrace
  \rho(\Acl) +\kappa\|\Delta\|\geq 1\right\rbrace$. The second equality follows from the fact that the random matrices $\Delta$ and $\widetilde \Delta$ are equal in distribution. The last inequality follows from \cite[Theorem.~4.1.1]{JAT:15}. 

\textit{(Lower bound)} From \eqref{eq: spectral_radius_bounds},
consider the set inclusion
$\{|\tr(\wAcl)|\geq n\}\subseteq \{\rho(\widetilde A_\text{cl})\geq
1\}$, which implies
$\bP[|\tr(\wAcl)|\geq n]\leq \mathbb{P}[ \rho(\wAcl)\geq 1]$. Further,
from \eqref{eq: perturbed_system} it follows that
$\tr(\widetilde A_\text{cl})=\mu+\sum_\mathrm{supp} (Z)
z_i\tr(BJ_i)$. Since $z_i\sim\mathcal{N}(0,\sigma_i^2)$ and the terms
$\mu$ and $\tr(BJ_i)$ are known scalars, $\tr(\wAcl)$ is distributed
according to $\mathcal{N}(\mu,\ul v)$. Hence, $|\tr(\wAcl)|$ follows a
folded normal distribution \citep{FCL-LSN-RBN:61}, and, by definition,
$\bP[|\tr(\wAcl)|\geq n]=\mathbb{Q}((n+\mu)/\sqrt{\ul
  v})+\mathbb{Q}((n-\mu)/\sqrt{\ul v})$.
\end{proof}

The bounds in Theorem \ref{thm: bound_probability_gaussian} quantify
how different properties of the nominal system dynamics and the data
perturbation affect the stability of the closed-loop dynamics. First,
the variance parameters $\ol v$ and $\ul v$ depend on the variance of
the perturbation ($\sigma_i$), the number of perturbed entries
($\mathrm{supp} (Z)$), and the sensitivity of the data-driven control
algorithm, as captured by the Jacobian matrices $J_i$. In particular,
when the variance of the perturbation grows and the other quantities
remain bounded, $\ul v$ grows to infinity and the lower bound in
\eqref{eq: probability_bounds} converges to $1$, since
$\mathbb{Q}(\cdot)$ converges to $0.5$. As intuitively expected, the
probability of having a stabilizing controller decreases to zero for
perturbations of increasing variance. Conversely, when the variance of
the perturbation, the number of perturbed experiments, or the
sensitivity of the data-driven algorithm converge to zero, then,
assuming the other quantities remain bounded, $\ol v$ decreases to
zero and the upper bound in \eqref{eq: probability_bounds} converges
to zero. This shows that the closed-loop system is stable with
probability growing to one when the effect of the perturbation on the
data-driven controller decreases to zero. Second, the eigenvalues and
the non-normality degree of the nominal closed-loop system, as
measured by the condition number $\kappa$ \citep{LNT-ME:05}, also
affect the performance of the data-driven controller. Specifically,
the upper bound in \eqref{eq: probability_bounds} grows with the
condition number $\kappa$, as expected since the sensitivity of the
eigenvalues of a matrix increases with its condition number
\citep{LNT-ME:05}, and with the spectral radius $\rho(\Acl)$, since
matrices with eigenvalues closer to the unit circle require smaller
perturbations to become unstable. Similarly, the lower bound in
\eqref{eq: probability_bounds} is also increasing with respect to
$|\mu|$, thus yielding a larger lower bound for nominal systems that
are closer to instability. Finally, Theorem \ref{thm:
  bound_probability_gaussian} can be used to characterize the rate at
which the probability of instability of the closed-loop system grows
as a function of the number of perturbed experiments, as we show next.

\begin{theorem}{\bf\emph{(Convergence rate)}}\label{cor: rate of
    increase}
  Let $i \in \mathrm{supp} (Z)$, and define
  $\mathrm{diag}(BJ_i)=[\gamma_{i}^1,\ldots,\gamma_{i}^n]$,
  $\alpha_i=\mathrm{min}\{\gamma_{i}^1,\ldots,\gamma_{i}^n\}$, and
  $\gamma = \min_i \{\sigma_i\alpha_i \}$. Then,
  $\mathbb{P}[ \rho(\wt{A}_{\mathrm{cl}})\geq 1] >
  2\mathbb{Q}(2/\sqrt{\gamma^2 |\mathrm{supp}(Z)|})$.
\end{theorem} 
\begin{proof}
  Because $|\mu|<n$ and $\mathbb{Q}(\cdot)$ is a monotone function,
  \eqref{eq: probability_bounds} implies that
  $2\mathbb{Q}(2n/\sqrt{\ul v})\leq \mathbb{P}[\rho(\widetilde
  A_\text{cl})\geq 1]$. The Theorem follows from
  $\sqrt{\ul v}\geq
  \sqrt{|\mathrm{supp}(Z)|\min_i[\mathrm{tr}(\sigma_iBJ_i)]^2}\geq
  \sqrt{|\mathrm{supp}(Z)| n^2\gamma^2}$.
\end{proof}

Since
$2\mathbb{Q}(2/\sqrt{\gamma^2 |\mathrm{supp}(Z)|}) =
1-\mathrm{erf}(1/\sqrt{0.5\,\gamma^2 |\mathrm{supp}(Z)|})$, Theorem
\ref{cor: rate of increase} states that
$\mathbb{P}[ \rho(\wt{A}_{\mathrm{cl}})\geq 1]$ increases to one at
the rate of a Gaussian error function of order
$1/\sqrt{|\mathrm{supp}(Z)|}$. Further, the convergence rate towards
instability is independent of the dimension of the closed-loop system.


To conclude this section, we discuss the performance of the
data-driven controller when the number or length of the experiments
grows and the number of perturbed entries remain bounded. In this
case, if the data-driven control algorithm $F$ depends in a
comparable way on all data points but not almost exclusively on any of
them, then the Jacobian matrices $J_i$ have decreasing norm, and the
upper bound in Theorem \ref{thm: bound_probability_gaussian} decreases
to zero. This implies that the data-driven algorithm becomes
increasingly more robust to perturbations that are bounded in variance
and support as the number of experimental data increases. To formalize
this discussion, let $\overline v$ be as in Theorem \ref{thm:
  bound_probability_gaussian}, and notice that
\begin{align}\label{eq: Jmax}
  \overline v \le \sigma_\text{max}^2  | \mathrm{supp} (Z)| 
  J_\text{max}^2
\end{align}
where $\sigma_\text{max} = \max_{i \in \mathrm{supp}(Z)} \sigma_i$ and
$J_\text{max} = \max_{i \in \mathrm{supp}(Z)} \| B J_i\|$. Then,
whenever $|\mathrm{supp} (Z)|  J_\text{max}^2$ decreases and
$\subscr{\sigma}{max}$ remains bounded, $\overline v$ converges to
zero, and Theorem \ref{thm: bound_probability_gaussian} implies that
the perturbed closed-loop system remains stable with probability
converging to one. This robustness property, which we validate in
Section \ref{sec: examples} for a class of data-driven control
algorithms, is in contrast to model-based control techniques, where
only a finite number of perturbations can in general be detected and
remedied (e.g., see \cite{SS-CH:10a,FP-FD-FB:10y}).




\begin{remark}{\bf\emph{(Tightness of the bounds)}} \label{remark:
    tightness of bounds} The bounds in \eqref{thm:
    bound_probability_gaussian} depends on the dimension of
  $A_\mathrm{cl}$. Although the lower bound ranges between $0$ and $1$,
  the upper bound can exceed $1$, as suggested by the factor $2n$
  outside the exponential function. Other factors can also deteriorate
  the upper bound; see \citep[Chapter.~4]{JAT:15} for a thorough
  discussion of the role of the dimension on probabilistic tail
  bounds. Yet, in addition to providing a qualitative understanding of
  the properties that affect closed-loop stability with perturbed
  data, the bounds in \eqref{thm: bound_probability_gaussian} remain
  useful in many cases (e.g., see Fig. \ref{fig: bounds}). \oprocend
\end{remark}

\begin{remark}{\bf\emph{(Gaussian assumption in Theorem \ref{thm:
        bound_probability_gaussian})}} The results in Theorem
  \ref{thm: bound_probability_gaussian} can be readily extended to
  different classes of stochastic perturbations. For instance, an
  upper bound similar to the one in \eqref{eq: probability_bounds} can
  be obtained for perturbations with bounded support (see Matrix
  Bernstein Inequality in \citep[Chapter.~6]{JAT:15}). Lower bounds
  can be obtained using the Paley-Zygmund or Cantelli's inequality,
  although such results would likely be loose without any further
  assumption on the perturbation. \oprocend
\end{remark}


\begin{figure}[t]
	\centering
	\includegraphics[width=0.8\columnwidth]{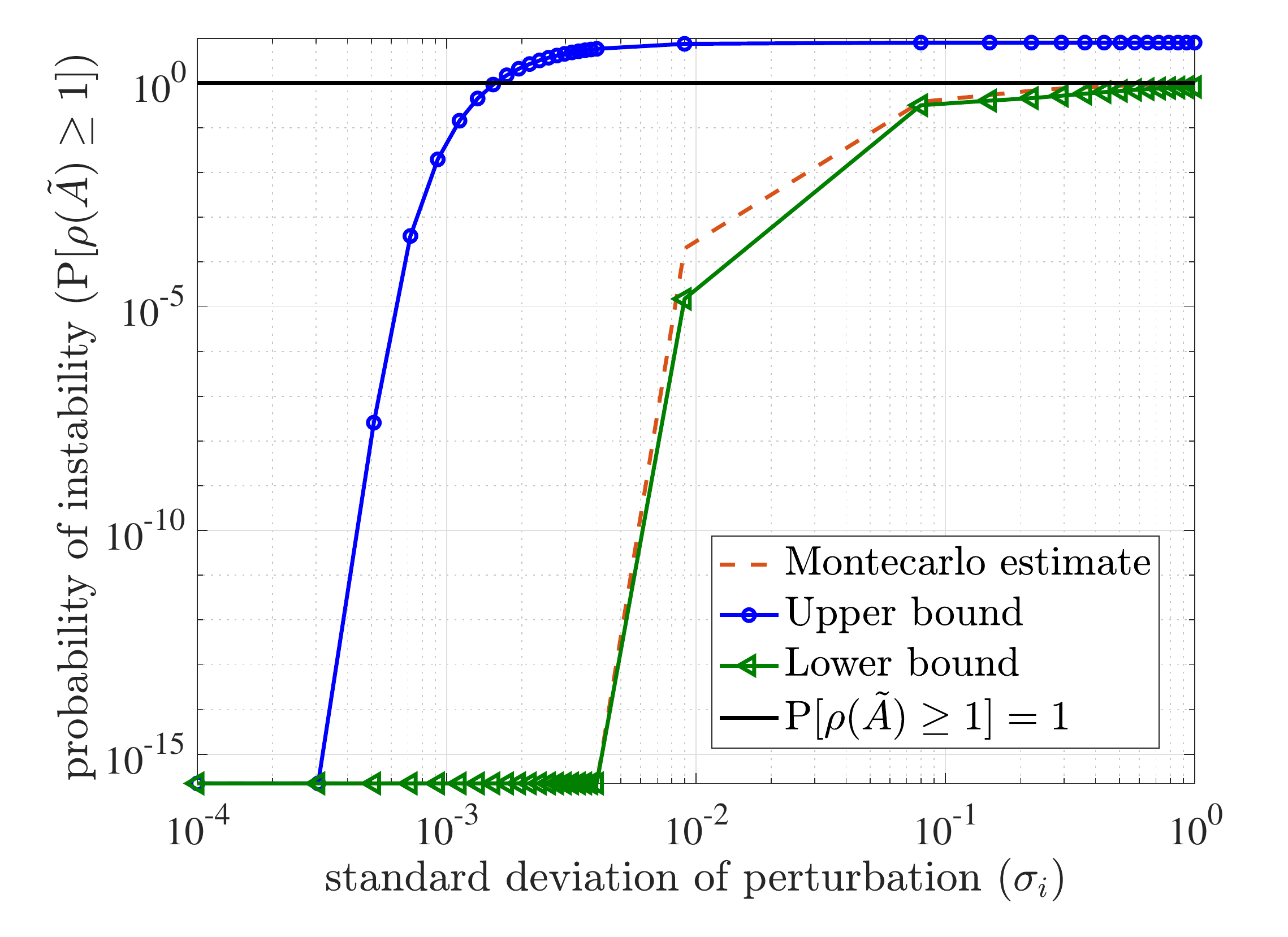}
	\caption{This figure shows the estimate of probability of
          instability (orange dashed line) of perturbed the
          closed-loop system described in \eqref{eq: example}, and the
          upper and lower bound (blue circle and green triangle,
          respectively) derived in Theorem \ref{thm:
            bound_probability_gaussian} as a function of perturbation
          variance ($\sigma_i$). Any value less than that of machine
          epsilon is rounded to that value $(2.2\mathrm{e}-16)$ Notice
          that (i) the probability of instability lies within the
          theoretical bounds derived in Theorem \ref{thm:
            bound_probability_gaussian}, (ii) the upper (resp. lower)
          bound converges to zero (resp. one) as the perturbation
          variance decreases (resp. increases) (iii) the lower bound
          is tight for all values of $\sigma_i$; however, the upper
          bound proves to be meaningful only when
          $\sigma_i \in (1\mathrm{e}-4, 1.5\mathrm{e}-3)$.}
	\label{fig: bounds}
\end{figure}

\section{Illustrative examples}\label{sec: examples}
In this section we provide examples to illustrate the bounds derived
in Theorem \ref{thm: bound_probability_gaussian}. To this aim, we
consider a simplified discrete-time linear time-invariant model of a
vehicle \citep{SD-NM-BR-VY:19}:
\begin{align}\label{eq: example}
x(t+1) =
\begin{bmatrix}
1 & T_s & 0 & 0 \\
0 & 1 & 0 & 0\\
0 & 0 & 1 & T_s\\
0 & 0 & 0 & 1
\end{bmatrix}
x(t) +
\begin{bmatrix}
0 & 0 \\
T_s & 0\\
0 & 0\\
0 & T_s
\end{bmatrix}
u(t),
\end{align}
where $x(t) \in \mathbb{R}^4$ contains the vehicle’s position and
velocity in cartesian coordinates, $u(t) \in \mathbb{R}^2$ is the
input signal, and $T_s=0.1$ is the sampling time. We assume that the
matrices in \eqref{eq: example} are unknown, and collect the state
trajectory resulting from a single control experiment with a random
control input of length $T = 500$, which implies that
$\mathrm{vec}(X)\in \mathbb{R}^{2000}$. Then, we use the data-driven
characterization provided in \cite[Theorem 3]{CDP-PT:19} as a
procedure to design a minimum-norm stabilizing controller for
\eqref{eq: example}. We let the experimental data be perturbed at $50$
random locations in $\mathrm{vec}(X)$, that is,
$|\mathrm{supp}(Z)|=50$, and compute the Montecarlo estimate of the
probability of instability (numerically, over $10000$ instances) of
the closed-loop system for different values of the variance of the
Gaussian perturbation (with zero mean). Our results are in
Fig. \ref{fig: bounds}. We remark that the Jacobian of the data-driven
control algorithm ($J_X$ in \eqref{eq: frechet_derivative}), and thus
the matrices $J_i$ in \eqref{eq: probability_bounds}, can be computed
numerically using the available training data, similarly to the
numerical computation of the derivative of a scalar function. We refer
the reader to \cite{MCS-YZH-CHW-PCW:17}.

\begin{figure}[t]
	\centering
	\includegraphics[width=0.65\columnwidth, trim = 0.1cm 0.1cm
        0.1cm 0.15cm, clip]{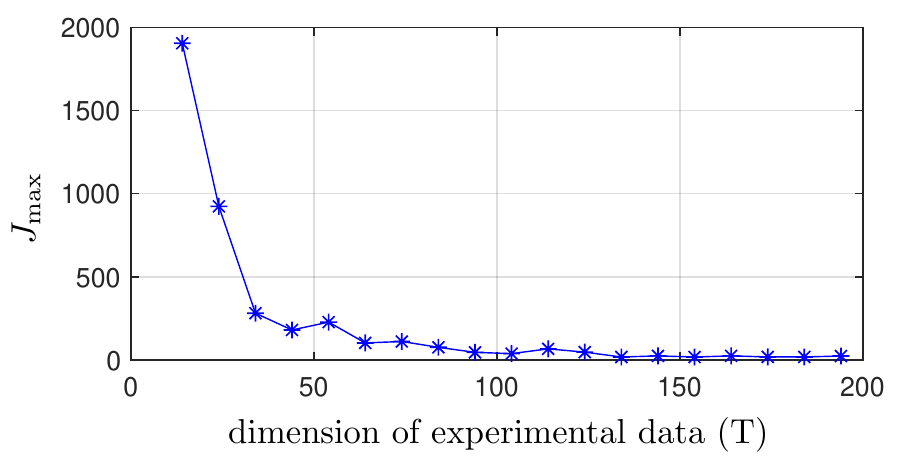}
	\caption{For the setup in Section \ref{sec: examples}, this
          figure shows the norm $\subscr{J}{max}$ in \eqref{eq: Jmax}
          as a function of the dimension of the experimental data
          (average over $15$ trials). The norm $\subscr{J}{max}$ is a decreasing
          function of the dimension of the training data, which
          ensures robustness when the dimension of the experimental
          data grows faster than the dimension of the compromised~data.}
	\label{fig: Jmax decreasing}
\end{figure}

To conclude, in Fig. \ref{fig: Jmax decreasing} we show that the
Jacobian matrix $J_X$ of the considered data-driven control algorithm
satisfies the bound in \eqref{eq: Jmax}. That is, the sensitivity of
the algorithm in \citep[Theorem 3]{CDP-PT:19} to variations of the
training data decreases with the dimension of the training data. This
ensures that localized perturbations have increasingly less effect on
the final feedback controller, and that the stability of the perturbed
closed-loop system is maintained with higher probability. We leave the
analytical characterization of this property as the subject of ongoing
and future investigation.


\section{Conclusion}
In this paper we describe a novel framework to quantify the robustness
of data-driven control algorithms for linear systems against
stochastic perturbations of the training data. We derive lower and
upper bounds for the probability of the spectral radius of the
closed-loop system exceeding one, as a function of the perturbation
statistics, sensitivity of the data-driven algorithm, and properties
of the nominal closed-loop system. We also characterize the rate at
which the probability of stability of the closed-loop system decreases
with the cardinality of the compromised data, and show that such rate
is independent of the system dimension. We discuss the qualitative
implications of our bounds, and show their effectiveness through
numerical simulations. Directions of future research include the
derivation of tighter bounds, especially upper bounds since our
estimate becomes increasingly more loose with the system dimension,
the generalization to more complex algorithms and perturbation models,
and the analysis of the sensitivity properties of different
data-driven control algorithms.

\bibliography{alias,FP,Main,New}

\end{document}